\documentclass[twocolumn,english,prl,reprint,nofootinbib,superscriptaddress,longbibliography]{revtex4-2}
\usepackage[T1]{fontenc}
\usepackage[latin9]{inputenc}
\usepackage{geometry}
\usepackage{color}
\usepackage{babel}
\usepackage{amsmath}
\usepackage{amsthm}
\usepackage{amssymb}
\usepackage{graphicx}
\usepackage[unicode=true,pdfusetitle,
 bookmarks=true,bookmarksnumbered=false,bookmarksopen=false,
 breaklinks=false,pdfborder={0 0 1},backref=false,colorlinks=true]
 {hyperref}
\hypersetup{
 pdfborderstyle=,allcolors=magenta}

\makeatletter
\theoremstyle{plain}
\newtheorem{thm}{\protect\theoremname}
\theoremstyle{plain}
\newtheorem{prop}[thm]{\protect\propositionname}
\theoremstyle{plain}
\newtheorem{lem}{\protect\lemmaname}

\usepackage{babel}
\usepackage{txfonts}
\usepackage{braket}

\providecommand{\propositionname}{Proposition}
\providecommand{\theoremname}{Theorem}
\providecommand{\lemmaname}{Lemma}

\makeatother

\providecommand{\propositionname}{Proposition}
\providecommand{\theoremname}{Theorem}

\begin{document}
\title{Quantum speed limits for change of basis}
\author{Moein Naseri}
\affiliation{Centre for Quantum Optical Technologies IRAU, Centre of New Technologies,
University of Warsaw, Poland}
\author{Chiara Macchiavello}
\affiliation{Dipartimento di Fisica, Universit\`a di Pavia, via Bassi 6, I-27100 Pavia, Italy}
\affiliation{INFN Sezione di Pavia, via Bassi 6, I-27100, Pavia, Italy}
\affiliation{CNR-INO, largo E. Fermi 6, I-50125, Firenze, Italy}
\author{Dagmar Bru\ss}
\affiliation{Institut f\"ur Theoretische Physik III, Heinrich-Heine-Universit\"at D\"usseldorf,~~\\
 D-40225 D\"usseldorf, Germany}
\author{Pawe\l{} Horodecki}
\affiliation{International Centre for Theory of Quantum Technologies,
University of Gda\'nsk, Wita Stwosza 63, 80-308 Gda\'nsk, Poland}
\affiliation{Faculty of Applied Physics and Mathematics, National Quantum Information Centre,
Gda\'nsk University of Technology, Gabriela Narutowicza 11/12, 80-233 Gda\'nsk, Poland}

\author{Alexander Streltsov}
\email{a.streltsov@cent.uw.edu.pl}
\affiliation{Centre for Quantum Optical Technologies IRAU, Centre of New Technologies,
University of Warsaw, Poland}

\begin{abstract}
Quantum speed limits provide ultimate bounds on the time required to transform one quantum state into another. Here, we extend the notion of quantum speed limits to collections of quantum states, investigating the time for converting a basis of states into an unbiased one. We provide tight bounds for systems of dimension smaller than 5, and general bounds for multi-qubit systems and Hilbert space dimension $d$. For two-qubit systems, we show that the fastest transformation implements two Hadamards and a swap of the qubits simultaneously. We further prove that for qutrit systems the evolution time depends on the particular type of the unbiased basis. We also investigate speed limits for coherence generation, providing the minimal time to establish a certain amount of coherence with a unitary evolution.
\end{abstract}

\maketitle

\textbf{\emph{Introduction.}} Striving for quantum advantages, such
as an increased speed of a computation, has become a competitive goal.
However, nature has established a fundamental speed limit, via a minimal
time that is necessary for the unitary evolution of an initial quantum
state to a final quantum state, as pointed out in~\cite{Mandelstam1945,Margolus1998188}.
In a geometric approach~\cite{JonesPhysRevA.82.022107,ZwierzPhysRevA.86.016101,PiresPhysRevX.6.021031,CampaioliPhysRevLett.120.060409},
the quantum speed limit is linked to the length of the shortest path
between initial and final state, which can be quantified via a suitable
distance measure. For a recent review of quantum speed limits, see~\cite{Deffner_2017}. 

The standard approach to quantum speed limits assumes that a quantum
state $\ket{\psi}$ is transformed into another state $\ket{\phi}$
via a unitary evolution $U=e^{-iHt}$. The task is to determine the
optimal evolution time for the transition $\ket{\psi}\rightarrow\ket{\phi}$,
with respect to the energy scale of the Hamiltonian $H$. First results
in this direction were presented for orthogonal states, and are known
as Mandelstam-Tamm bound~\cite{Mandelstam1945}: 
\begin{equation}
T_{\perp}\geq\frac{\pi}{2\Delta E_{\psi}}\;,\label{eq:Mandelstam}
\end{equation}
where $(\Delta E_{\psi})^{2}=\braket{H^{2}}_{\psi}-\braket{H}_{\psi}^{2}$
is the energy variance. Another bound was derived later by Margolus
and Levitin~\cite{Margolus1998188}, giving 
\begin{equation}
T_{\perp}\geq\frac{\pi}{2E_{\psi}}\;,\label{eq:Margolus}
\end{equation}
with the mean energy $E_{\psi}=\braket{H}_{\psi}-E_{0}$, and $E_{0}$
is the ground state energy. Note that the speed limits~(\ref{eq:Mandelstam})
and~(\ref{eq:Margolus}) differ only by the different choice of the
energy scale. For transition between mixed states $\rho\rightarrow\sigma$
generalized quantum speed limits have been presented~\cite{LevitinPhysRevLett.103.160502,PiresPhysRevX.6.021031,CampaioliPhysRevLett.120.060409,ShanahanPhysRevLett.120.070401}:
\begin{equation}
T(\rho\rightarrow\sigma)\geq\frac{\arccos F(\rho,\sigma)}{\min\left\{ \Delta E_{\rho},E_{\rho}\right\} }\label{eq:MandelstamLevitinMixed}
\end{equation}
with fidelity $F(\rho,\sigma)=\mathrm{Tr}\sqrt{\sqrt{\rho}\sigma\sqrt{\rho}}$. 

While the original approaches~\cite{Mandelstam1945,Margolus1998188} studied the speed limit for unitary transitions between two quantum states,  more general versions of the speed limit have been developed in the last years. This includes investigation of quantum speed limits for open system dynamics~\cite{delCampoPhysRevLett.110.050403,Funo_2019,Teittinene23030331,Teittinen_2019}, as well as speed limits for the evolution of observables in the Heisenberg picture~\cite{MohanArxiv.2112.13789}, and the study of speed limit for a bounded energy spectrum~\cite{PhysRevLett.129.140403}. A theoretical approach for measuring quantum speed limits in an ultracold gas has been proposed recently in~\cite{delCampoPhysRevLett.126.180603}. Speed limits for generating quantum resources have also been considered~\cite{Campaioli_2022}, allowing to determine optimal rates for generating quantum entanglement~\cite{HorodeckiRevModPhys.81.865}, quantum coherence~\cite{StreltsovRevModPhys.89.041003}, and quantum discord~\cite{ModiRevModPhys.84.1655,Streltsov_2015}.

Note that the early approaches~\cite{Mandelstam1945,Margolus1998188} studied the speed limit for transforming
\emph{one} quantum state into another one. However, many quantum technological
applications require to transform a collection of states. An important
example is quantum computation where a common operation is a change
of basis, e.g. by applying the well-known Hadamard gate which transforms
the computational qubit basis $\{\ket{0},\ket{1}\}$ into $\{\ket{+},\ket{-}\}$,
with $\ket{\pm}=(\ket{0}\pm\ket{1})/\sqrt{2}$. 

\begin{figure} 
\includegraphics[width=1\columnwidth]{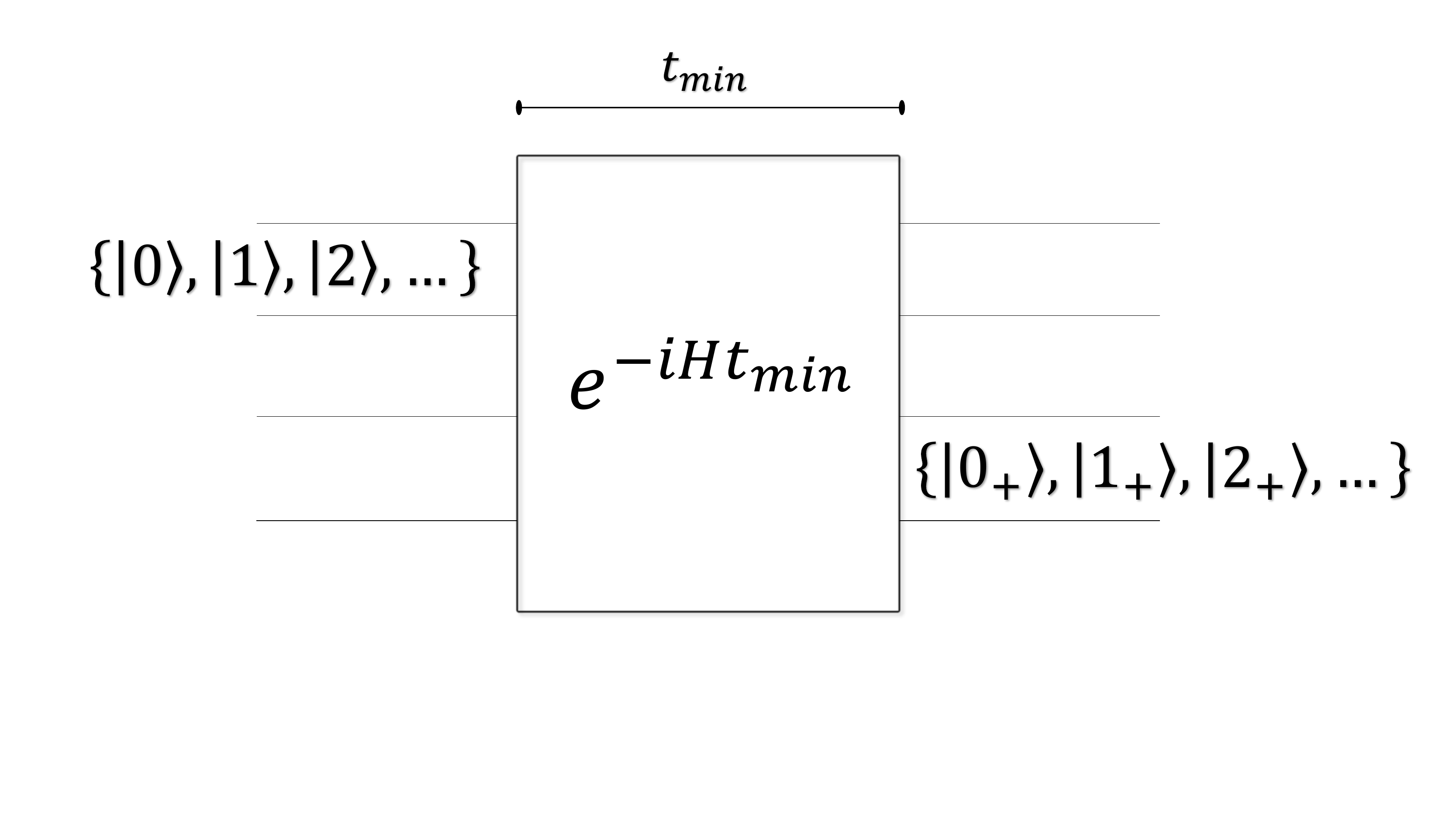}

\caption{\label{fig:SpeedLimitUnbiased} Generation of an unbiased basis $\{\ket{n_+}\}$ from the computational basis $\{\ket{n}\}$ via a unitary evolution $e^{-iHt_{\min}}$.}
\end{figure}

Which fundamental speed limits hold for such a basis transformation?
We address this question in this Letter, investigating bounds
on the time that is necessary to perform a basis change, i.e. a transformation
of an ordered set of quantum states to another ordered set of quantum
states, minimized over all Hamiltonians. In the spirit of the Margolus-Levitin
bound~(\ref{eq:Margolus}), we aim for quantum speed limits of the
form 
\begin{equation}
T(\ket{\psi_{j}}\rightarrow\ket{\phi_{j}})\geq\frac{g}{E}.\label{eq:GeneralSpeedLimit}
\end{equation}
Here $\{\ket{\psi_j}\},\{\ket{\phi_j}\}$
are two ordered sets of orthonormal basis states, with $j=1,...,d$,
where $d$ is the dimension of the Hilbert space,
and $g$ can in general depend on the sets $\{\ket{\psi_j}\}$ and $\{\ket{\phi_j}\}$. The quantity $E$ in Eq.~(\ref{eq:GeneralSpeedLimit})
denotes the mean energy of the Hamiltonian, which we define as 
\begin{equation}
E=\frac{1}{d}\sum_{j}\braket{\psi_{j}|H|\psi_{j}} - E_{0},\label{eq:MeanEnergy}
\end{equation}
naturally generalizing the mean energy $E_{\psi}$ appearing in the
Margolus-Levitin bound~(\ref{eq:Margolus}). The mean energy~(\ref{eq:MeanEnergy})
is equivalent to $E=\mathrm{Tr}[H/d] - E_{0}$, and thus independent
on the particular choice of basis $\{\ket{\psi_{j}}\}$. We also note that the mean energy is additive for non-interactive Hamiltonians of the form $H^{AB} = H^A \otimes \openone^B + \openone^A \otimes H^B$: 
\begin{equation}
    E_{AB} = E_A + E_B,
\end{equation}
where $E_A$ and $E_B$ are the mean energies of $H^A$ and $H^B$, respectively.

In addition to investigating speed limits for the change of basis, we also study speed limits for coherence generation. In particular, we consider the maximal coherence which can be established within a certain time, given some Hamiltonian with mean energy $E$. These results are highly relevant in the context of the resource theory of quantum coherence~\cite{BaumgratzPhysRevLett.113.140401,WinterPhysRevLett.116.120404,StreltsovRevModPhys.89.041003}, taking into account that several recent works suggest that quantum coherence is more suitable than entanglement to capture the performance of certain quantum algorithms~\cite{Matera_2016,AhnefeldPhysRevLett.129.120501,Naseri2022}. 

\medskip{}

\textbf{\emph{Speed limits for unbiased bases.}} In the following,
we will determine speed limits for basis change from the computational
basis $\{\ket{n}\}$ into an unbiased basis $\{\ket{n_{+}}\}$ with
$|\!\braket{n|n_{+}}\!|^{2}=1/d$, see also Fig.~\ref{fig:SpeedLimitUnbiased}. For single-qubit systems we obtain the bound
\begin{align}
T_{\mathrm{unbiased}} & \geq\frac{\pi}{4E}, \label{eq:UnbiasedBasisQubit}
\end{align}
which is tight for any unbiased qubit basis. See Appendix~\ref{sec:SingleQubit}
for more details on speed limits for single-qubit transitions.

It is now intuitive to assume that for $d > 2$ the evolution time into an unbiased basis increases, compared to the qubit setting. To support this intuition, consider a two-qubit system $AB$, and let $H^A$ and $H^B$ be qubit Hamiltonians which bring $\{\ket{0},\ket{1}\}$ into $\{\ket{+},\ket{-}\}$ within minimal time $\pi/(4E_A)$ and $\pi/(4E_B)$, respectively. If we set $E_A = E_B$, the Hamiltonian $H^{AB} = H^A \otimes \openone^B + \openone^A \otimes H^B$ achieves the transformation
\begin{equation}
\left\{ \ket{00},\ket{01},\ket{10},\ket{11}\right\} \,\,\,\,\rightarrow\,\,\,\,\{\ket{++},\ket{+-},\ket{-+},\ket{--}\}
\end{equation}
within time $\pi/(4E_A) = \pi/(2E)$, where $E=2E_A$ is the mean energy of the total Hamiltonian $H^{AB}$. From this argument, we see that for $d=4$ an unbiased basis can be achieved within time $\pi/(2E)$, which is longer compared to the single-qubit setup.

As we will see in the following, this intuition is not correct. For this, we will first focus on qutrit systems. As we show in
Appendix~\ref{sec:Unbiased-Qutrits}, a general unbiased qutrit basis
can be obtained via a diagonal unitary 
\begin{equation}
V=\sum_{j}e^{i\alpha_{j}}\ket{j}\!\bra{j}
\end{equation}
from one of the following two bases (denoted by $\{\ket{n_{+}}\}$
and $\{\ket{\tilde{n}_{+}}\}$, respectively): \begin{subequations}\label{eq:Unbiased-1}
\begin{align}
\ket{0_{+}} & =\frac{1}{\sqrt{3}}\left(\ket{0}+e^{i\frac{2}{3}\pi}\ket{1}+e^{i\frac{4}{3}\pi}\ket{2}\right),\\
\ket{1_{+}} & =\frac{1}{\sqrt{3}}\left(\ket{0}+\ket{1}+\ket{2}\right),\\
\ket{2_{+}} & =\frac{1}{\sqrt{3}}\left(\ket{0}+e^{-i\frac{2}{3}\pi}\ket{1}+e^{-i\frac{4}{3}\pi}\ket{2}\right),
\end{align}
\end{subequations} and \begin{subequations}\label{eq:Unbiased-2}
\begin{align}
\ket{\tilde{0}_{+}} & =\frac{1}{\sqrt{3}}\left(\ket{0}+e^{-i\frac{2}{3}\pi}\ket{1}+e^{-i\frac{4}{3}\pi}\ket{2}\right),\\
\ket{\tilde{1}_{+}} & =\frac{1}{\sqrt{3}}\left(\ket{0}+\ket{1}+\ket{2}\right),\\
\ket{\tilde{2}_{+}} & =\frac{1}{\sqrt{3}}\left(\ket{0}+e^{i\frac{2}{3}\pi}\ket{1}+e^{i\frac{4}{3}\pi}\ket{2}\right).
\end{align}
\end{subequations} 
Note that these two sets of basis states are odd permutations of each other. As discussed in Appendix~\ref{sec:ProofV}, this
implies that speed limits for the transitions $\{\ket{n}\}\rightarrow\{\ket{n_{+}}\}$
and $\{\ket{n}\}\rightarrow\{\ket{\tilde{n}_{+}}\}$ will also lead
to speed limits for general unbiased qutrit bases $\{\ket{n}\}\rightarrow\{V\ket{n_{+}}\}$
and $\{\ket{n}\}\rightarrow\{V\ket{\tilde{n}_{+}}\}$ with a diagonal
unitary $V$. Equipped with these tools, we now present the first
main result of this Letter. 
\begin{thm}
\label{thm:UnbiasedQutrit}The time for converting a qutrit basis
onto an unbiased basis is bounded below as 
\begin{equation}
T_{\mathrm{unbiased}}\geq\frac{2\pi}{9E}.\label{eq:UnbiasedQutritBound}
\end{equation}
\end{thm}
\noindent We refer to Appendix~\ref{sec:Proof-Theorem1} for the
proof.

Having established a speed limit for basis change it is natural to
ask whether this bound is tight, i.e., whether for any unbiased basis
there exists a Hamiltonian $H$ with mean energy $E$ saturating the
bound~(\ref{eq:UnbiasedQutritBound}). Recalling the definition of
the unbiased bases $\{\ket{n_{+}}\}$ and $\{\ket{\tilde{n}_{+}}\}$
in Eqs.~(\ref{eq:Unbiased-1}) and~(\ref{eq:Unbiased-2}), we answer
this question in the following proposition. 
\begin{prop} \label{prop:tight}
The speed limit~(\ref{eq:UnbiasedQutritBound}) is tight for the
basis $\{\ket{n_{+}}\}$, but not tight for basis $\{\ket{\tilde{n}_{+}}\}$. 
\end{prop}

\noindent We refer to Appendix~\ref{sec:Proof-Proposition2} for
the proof.

The above results imply that there are two different classes of unbiased
bases for qutrits: bases of the form $\{V\ket{n_{+}}\}$ can be obtained
from the computational basis at time $T=2\pi/9E$, while bases of
the form $\{V\ket{\tilde{n}_{+}}\}$ require an evolution time $T>2\pi/9E$,
where $V$ is an arbitrary diagonal unitary. For the second class
$\{V\ket{\tilde{n}_{+}}\}$ we have numerical evidence that a tight
speed limit is given as 
\begin{equation}
T\left(\ket{n}\rightarrow\ket{\tilde{n}_{+}}\right)\geq\frac{4\pi}{9E}.\label{eq:NumericalQutrit}
\end{equation}
To see this, note that any unitary achieving the transformation $\ket{n}\rightarrow\ket{\tilde{n}_{+}}$
must be of the form 
\begin{equation}
U=\sum_{n=0}^{2}e^{i\phi_{n}}\ket{\tilde{n}_{+}}\!\bra{n}\label{eq:Uqutrit-1}
\end{equation}
with some phases $\phi_{n}$ (see also Appendix~\ref{sec:Proof-Theorem1}).
Let now $\lambda_{j}=e^{-i\alpha_{j}}$ be the eigenvalues of $U$,
such that the phases $\alpha_{j}$ are in increasing order and $-\pi\leq\alpha_{j}\leq\pi$. For a given set of such phases $\{\alpha_j\}$, there exists a Hamiltonian implementing the unitary $U = e^{-iHt}$ such that 
\begin{equation}
    E_j t = \alpha_j \,\,\,\, \mathrm{or} \,\,\,\, E_j t = \alpha_j +2 \pi, \label{eq:NumericalTest-2}
\end{equation}
where $E_j$ are the eigenvalues of $H$. The mean energy of the numerically obtained Hamiltonian then fulfills 
\begin{equation}
    Et = \frac{1}{3} \sum_j E_j t - E_0t. \label{eq:NumericalTest}
\end{equation}
Using these results, we can test Eq.~(\ref{eq:NumericalQutrit}), by
numerically sampling random phases $0\leq\phi_{n}\leq2\pi$ and evaluating
$Et$ via Eq.~(\ref{eq:NumericalTest}). The choice of $E_jt$ as in Eq.~(\ref{eq:NumericalTest-2}) guarantees that the numerical Hamiltonians obtained in this way contain Hamiltonians with the minimal value of $Et$.

\begin{figure}
\includegraphics[width=1\columnwidth]{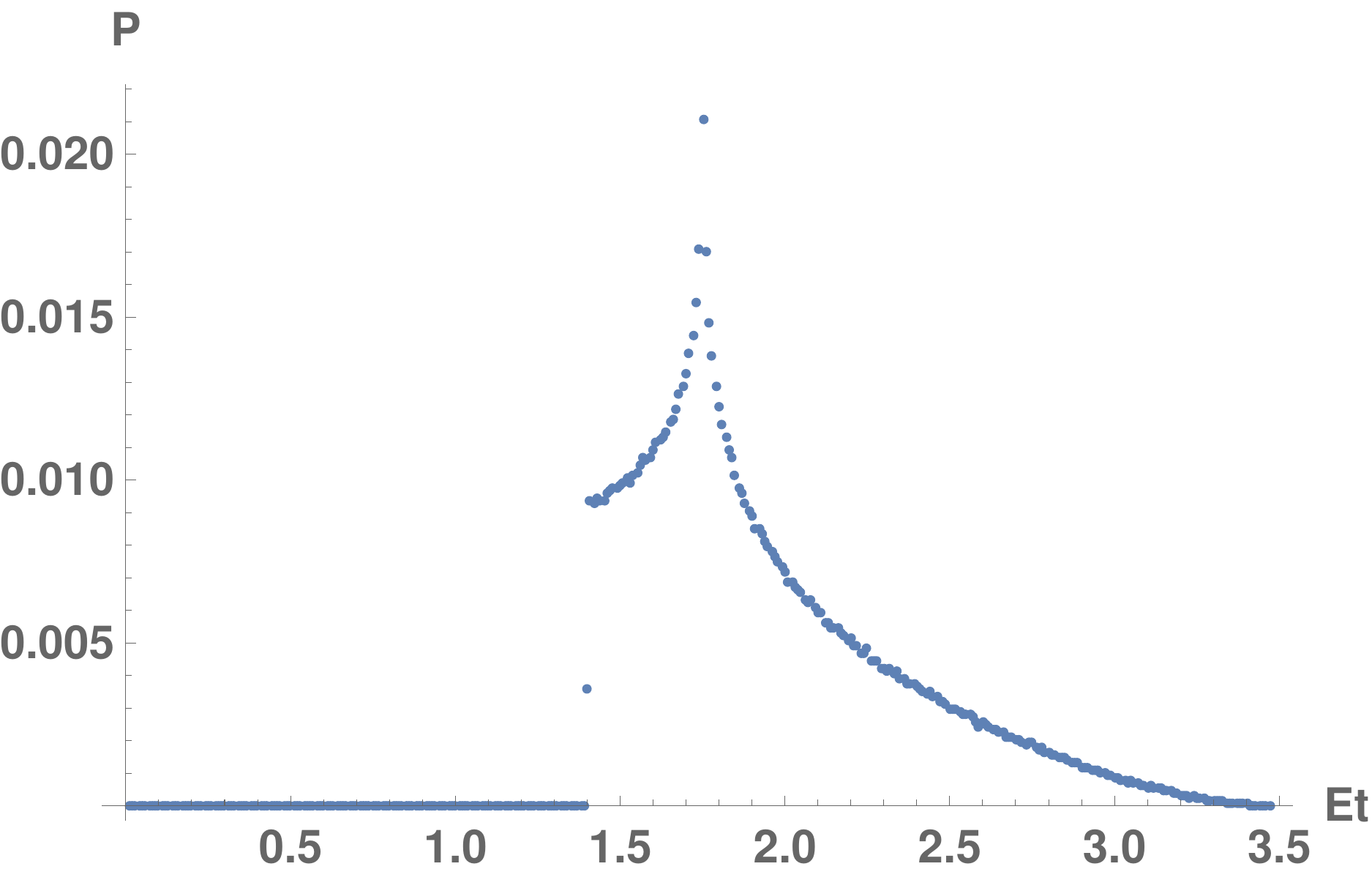}

\caption{\label{fig:QutritNumerics} Numerical test of Eq.~(\ref{eq:NumericalQutrit}).
We sample $10^{6}$ unitaries of the form~(\ref{eq:Uqutrit-1}) with
random phases $0\protect\leq\phi_{n}\protect\leq2\pi$ and evaluate
$Et$ using Eq.~(\ref{eq:NumericalTest}). The plot shows the numerical
probability as a function of $Et$. As a numerical bound, we obtain $Et \geq \frac{4}{9}\pi + \varepsilon$ with $\varepsilon \leq 10^{-5}$, in good agreement with Eq.~(\ref{eq:NumericalQutrit}).}
\end{figure}
In Fig.~\ref{fig:QutritNumerics}
we show the numerical probability for obtaining a certain value of
$Et$ for $10^{6}$ samples. The numerical results suggest the following
lower bound for $Et$: 
\begin{equation}
Et\geq\frac{4}{9}\pi + \varepsilon,\label{eq:NumericalBound}
\end{equation}
where $\varepsilon$ is numerically upper bounded as $\varepsilon \leq 10^{-5}$, in good agreement with Eq.~(\ref{eq:NumericalQutrit}). A Hamiltonian saturating the bound~(\ref{eq:NumericalQutrit}) is given by $\tilde{H}=-\ket{\tilde{\alpha}}\!\bra{\tilde{\alpha}}$
with 
\begin{equation}
\ket{\tilde{\alpha}}=\frac{1}{\sqrt{3}}(\ket{0}+e^{i\frac{2}{3}\pi}\ket{1}+\ket{2}).
\end{equation}

A direct comparison of Theorem~\ref{thm:UnbiasedQutrit} with the corresponding qubit bound (\ref{eq:UnbiasedBasisQubit}) shows that establishing an unbiased qutrit basis requires less time, compared to an unbiased qubit basis for the same mean energy $E$. In the following, we will discuss the main differences between the qubit and the qutrit setting.

If a single-qubit
unitary $U=e^{-iHt}$ is optimal for rotating the basis $\{\ket{0},\ket{1}\}$
onto an unbiased basis, then the unitary $U^{2}=e^{-2iHt}$ permutes
the basis elements $\{\ket{0},\ket{1}\}$. This is no longer the case
in the qutrit setting. For this, note that an optimal Hamiltonian
for the qutrit transition $\ket{n_{+}}=e^{-iHt}\ket{n}$ is given
by $H=\ket{\alpha}\!\bra{\alpha}$, with 
\begin{equation}
\ket{\alpha}=\frac{1}{\sqrt{3}}\left(\ket{0}+e^{-i\frac{2}{3}\pi}\ket{1}+\ket{2}\right).
\end{equation}
For the optimal Hamiltonian we can evaluate the fidelity between the
initial state $\ket{0}$ and the time-evolved state $e^{-iHt}\ket{0}$:
\begin{equation}
|\!\braket{0|e^{-iHt}|0}\!|^{2}=\frac{1}{9}[5+4\cos(t)].\label{eq:Fidelity}
\end{equation}
Note that the right-hand side of Eq.~(\ref{eq:Fidelity}) is never zero, which means that the evolution never permutes $\ket{0}$ with
another basis element, and the same can be shown for the states $\ket{1}$
and $\ket{2}$.

Moreover, if the single-qubit unitary $U$ permutes the basis states
$\{\ket{0},\ket{1}\}$, then $\sqrt{U}$ always rotates the $\{\ket{0},\ket{1}\}$
basis onto an unbiased basis. This is no longer the case in the qutrit
setting, as can be seen by inspection, with the permutation $U=\sum_{n=0}^{2}\ket{(n+1)\,\,\!\!\!\!\!\mod\!3}\!\bra{n}$.
We further obtain 
\begin{equation}
\sqrt{U}=\frac{1}{3}\left(\begin{array}{ccc}
2 & -1 & 2\\
2 & 2 & -1\\
-1 & 2 & 2
\end{array}\right),
\end{equation}
and thus $\sqrt{U}\ket{n}$ is not a maximally coherent state for any $0\leq n\leq2$. It can be verified by inspection that also $U^{1/3}$ does not transform any of the states $\ket{n}$ into a maximally coherent state.

So far, we considered systems of dimension $2$ and $3$. We will now go one step further, giving the minimal evolution time for an unbiased basis for two-qubit systems.
\begin{thm} \label{thm:UnbiasedTwoQubits}
The time for establishing an unbiased two-qubit basis is bounded below as 
\begin{equation}
    T_{\mathrm{unbiased}} \geq \frac{\pi}{4E}. \label{eq:UnbiasedTwoQubits}
\end{equation}
There exists a two-qubit Hamiltonian achieving this bound.
\end{thm}

\noindent Remarkably, this bound is the same as for single-qubit systems, see Eq.~(\ref{eq:UnbiasedBasisQubit}). The Hamiltonian saturating Eq.~(\ref{eq:UnbiasedTwoQubits}) is given as 
\begin{equation}
    H = -\sigma_x \otimes \sigma_z + \sigma_y \otimes \sigma_y - \sigma_z \otimes \sigma_x. \label{eq:TwoQubitHamiltonian}
\end{equation}
The eigenvalues of this Hamiltonian are $3$, $-1$, $-1$, $-1$, and the mean energy of $H$ is given as $E = 1$. For $t = \pi/4$ we now define the unitary $U = e^{-itH}$. The action of this unitary onto the computational basis of two qubits is as follows:
\begin{subequations}
\begin{align}
    U (\ket{0}\ket{0}) &= e^{i\pi/4} \ket{+} \ket{+}, \\
    U (\ket{0}\ket{1}) &= e^{i\pi/4} \ket{-} \ket{+}, \\
    U (\ket{1}\ket{0}) &= e^{i\pi/4} \ket{+} \ket{-}, \\
    U (\ket{1}\ket{1}) &= e^{i\pi/4} \ket{-} \ket{-}.
\end{align}
\end{subequations}
This shows that the Hamiltonian in Eq.~(\ref{eq:TwoQubitHamiltonian}) indeed transforms a two-qubit basis onto an unbiased basis within time $\pi/(4E)$. We refer to Appendix~\ref{sec:ProofUnbiasedTwoQubits} for the proof of Theorem~\ref{thm:UnbiasedTwoQubits} and more details.

The results presented so far show that the optimal time for transformation onto an unbiased basis is the same for single-qubit and two-qubit systems, and in both cases given by $\pi/(4E)$. For a qutrit system we have a shorter time $2\pi/(9E)$. We will now extend these results to many-qubit systems. As we will see, there exists a universal bound for $n$-qubit systems, allowing us to establish an unbiased basis within finite time.

\begin{thm} \label{thm:nQubits}
For systems with $n$ qubits, the minimal time for estabishing an unbiased basis is bounded above as 
\begin{equation}
    T_\mathrm{unbiased} \leq \frac{\pi}{2E}.
\end{equation}
\end{thm}
\begin{proof}
Consider the $n$ qubit Hamiltonian
\begin{equation}
    H_n = V^{\otimes n},
\end{equation}
where $V$ is the Hadamard gate. Note that the mean energy of $H_n$ is given as $E = 1$. We now define the unitary $U_n(t) = e^{-iH_nt}$. Using the fact that $H_n^2 = \openone$ it follows that 
\begin{equation}
    U_n(t) = \cos(t) \openone - i \sin(t) H_n.
\end{equation}
For $t=\pi/2$ we obtain  
\begin{equation}
    U_n(\pi/2) = -i V^{\otimes n}. \label{eq:Hadamard}
\end{equation}
This unitary transforms the computational basis of $n$ qubits into an unbiased basis, and the proof is complete.
\end{proof}

Theorem~\ref{thm:nQubits} shows that it is possible to establish an unbiased basis of $n$ qubits within time $\pi/(2E)$. We demonstrated this explicitly by presenting a Hamiltonian, which introduced interactions between all the qubits. Without interactions, i.e., if each of the qubits evolves independently, the optimal evolution time is given by $n\pi/(4E)$.

In the following, we present a general lower bound for the time required for establishing an unbiased basis for any $d$-dimensional system.

\begin{thm} \label{thm:LowerBoundGeneral}
The time for establishing an unbiased basis for a system of dimension $d$ is bounded below by 
\begin{equation}
T_{\mathrm{unbiased}}>\frac{\pi(d-1)}{4Ed}.
\end{equation}
\end{thm}
\noindent As we see, for large Hilbert space dimension the lower bound converges to $\pi/4E$. We refer to Appendix~\ref{sec:ProofLowerBound} for the proof of the theorem. For systems of dimension $6$ this bound can be improved slightly to $T \geq 0.227/E$, see Appendix \ref{sec:ProofLowerBound} for more details. Comparing this lower bound with the bound in the Theorem \ref{thm:nQubits}, we see that in the limit $n \rightarrow \infty$ the minimal time $T$ for establishing an unbiased basis of $n$ qubits fulfills $\pi/4E \leq T \leq \pi/2E$.

\medskip{}

\textbf{\emph{Speed limits for basis permutation.}} It is instrumental
to compare the above results to the speed limits for permuting the
basis $\{\ket{n}\}$: 
\begin{equation}
U\ket{n}=\ket{(n+1)\!\!\!\!\!\mod\!d}\label{eq:Uperm-1}
\end{equation}
for all $0\leq n\leq d-1$. 
\begin{prop}
The time for permuting a basis is bounded below by 
\begin{equation}
T_{\mathrm{perm}}\geq\frac{\pi(d-1)}{dE}.
\end{equation}
\end{prop}
\begin{proof}
As we discuss in the Appendix~\ref{sec:Permutation}, the eigenvalues
of the permutation unitary~(\ref{eq:Uperm-1}) have the form 
\begin{equation}
\lambda_{j}=e^{-i\frac{2\pi j}{d}},
\end{equation}
where integer $j$ is in the range $0\leq j\leq d-1$. It follows
that for any permutation unitary $U=e^{-iHt}$ it must hold that 
\begin{equation}
t\sum_{j}E_{j}=\sum_{j}\frac{2\pi j}{d}=\pi(d-1).
\end{equation}
The proof of the proposition is complete by noting that $E=\sum_{j}E_{j}/d$. 
\end{proof}
Interestingly, for a given Hamiltonian $H$ there are only two options:
either the unitary $U=e^{-iHt}$ leads to permutation with $t=\pi(d-1)/(dE)$,
or the Hamiltonian never leads to a basis permutation. We further
note that our analysis applies only to permutations of the form~(\ref{eq:Uperm-1}).

\medskip{}

\textbf{\emph{Speed of evolution for coherence generation.}} We will
now present speed limits for the creation of quantum coherence under
unitary evolution. In particular, we are interested in the maximal
value of coherence $C_{\max}$ which can be achieved from a given
state $\rho$ within a fixed time~$t$: 
\begin{equation}
C_{\max}(\rho,t)=\max_{H}\left\{ C\left(e^{-iHt}\rho e^{iHt}\right)\right\} ,
\end{equation}
and the maximization is performed over all Hamiltonians $H$ with
average energy $E=\mathrm{Tr}[H]/d - E_{0}$. As a quantifier of coherence
we use the $\ell_{1}$-norm of coherence~\cite{BaumgratzPhysRevLett.113.140401,StreltsovRevModPhys.89.041003}
\begin{equation}
C(\rho)=\sum_{i\neq j}|\rho_{ij}|,\label{cl1}
\end{equation}
which can be estimated efficiently in experiments by using collective measurements~\cite{Yuan2020,Wu2021}.

We will first discuss the single-qubit setting. Recall that in this
case the unitary $U(t)=e^{-iHt}$ can be interpreted as a rotation
by an angle $2Et$ about the axis $\boldsymbol{n}$ of the Bloch sphere.
As for single-qubit states the amount of coherence $C$ corresponds
to the Euclidean distance to the incoherent axis, $C_{\max}(\rho,t)$
corresponds to the largest distance from the incoherent axis, maximized
over all rotations with a fixed angle $2Et$. The optimal rotation
axis $\boldsymbol{n}$ is orthogonal to the Bloch vector $\boldsymbol{r}$
and the incoherent axis, and $C_{\max}$ takes the following form:
\begin{equation}
C_{\max}(\rho,t)=|\boldsymbol{r}|\cos\left(\arcsin\left[\frac{|r_{z}|}{|\boldsymbol{r}|}\right] - 2Et\right).
\end{equation}
Note that $C_{\max}$ cannot be larger than $|\boldsymbol{r}|$, and
this value is attained for the time 
\begin{equation}
T_{\mathrm{mc}}=\frac{1}{2E}\arcsin\frac{|r_{z}|}{|\boldsymbol{r}|},
\end{equation}
in which case the final state is in the maximally coherent plane.
If the initial state is pure, it can be parametrised as 
\begin{equation}
\ket{\psi}=\cos(\theta/2)\ket{0}+e^{i\phi}\sin(\theta/2)\ket{1}\;,\label{psi0}
\end{equation}
and the maximal amount of coherence achievable in a given time $t$
takes the form 
\begin{equation}
C_{\max}(\ket{\psi},t)=\cos\left(\arcsin\left[\cos\theta\right] - 2Et\right).
\end{equation}

In the next step we will consider systems of arbitrary dimension $d\geq2$
and evaluate the minimal time for converting a pure state $\ket{\psi}$
into a maximally coherent state of the form 
\begin{equation}
\ket{+}_{d}=\frac{1}{\sqrt{d}}\sum_{j=0}^{d-1}e^{i\phi_{j}}\ket{j}\label{eq:maximallycoherent}
\end{equation}
with phases $\phi_{j}$. The following proposition gives a bound for
the evolution time $T(\ket{\psi}\rightarrow\ket{+}_{d})$. 
\begin{prop}
The time for converting a state $\ket{\psi}$ into a maximally coherent
state $\ket{+}_{d}$ via unitary evolution $U=e^{-iHt}$ is bounded
as
\begin{equation}
T(\ket{\psi}\rightarrow\ket{+}_{d})\ge\frac{1}{dE}\arccos\left[\frac{2}{d}\left(\sum_{j}|\!\braket{\psi|j}\!|\right)^{2}-1\right]. \label{eq:PureMC}
\end{equation}
\end{prop}

\begin{proof}
From Lemma~\ref{lem:PureStates} in Appendix~\ref{sec:PureStates}, it follows that the evolution time into a maximally coherent state is bounded as 
\begin{equation}
T\left(\ket{\psi}\rightarrow\ket{+}_d\right)\geq\frac{1}{dE}\arccos\left(2|\!\braket{\psi|+}_d\!|^{2}-1\right). \label{eq:MCspeedlimitProof}
\end{equation}
Thus, in order to obtain a bound which is valid for all maximally coherent states, we need to estimate the maximal overlap $|\!\braket{\psi|+}_d\!|$ over all states of the form~(\ref{eq:maximallycoherent}). Expanding the initial state $\ket{\psi}$ in the incoherent basis
$\{\ket{i}\}$ as 
\begin{equation}
\ket{\psi}=\sum_{j=0}^{d-1}c_{j}e^{i\alpha_{j}}\ket{j}
\end{equation}
with $c_{j}\geq0$, it is straightforward to see that the overlap
$|\!\braket{\psi|+}_{d}\!|^{2}$ is maximized if we set $\phi_{j}=\alpha_{j}$,
thus arriving at 
\begin{equation}
\max_{\ket{+}_d}|\!\braket{\psi|+}_{d}\!|^{2}=\frac{1}{d}\left(\sum_{j}|\!\braket{\psi|j}\!|\right)^{2}. \label{eq:MCoverlap}
\end{equation}
Alternatively, this result can be obtained following~\cite{RegulaPhysRevLett.121.010401,RegulaPhysRevA.98.052329}, noting that $\max_{\ket{+}_d}|\!\braket{\psi|+}_{d}\!|^{2}$ corresponds to the maximal fidelity between the state $\Lambda[\ket{\psi}\!\bra{\psi}]$ and the particular maximally coherent state $\ket{+}_d = \sum_j \ket{j}/\sqrt{d}$, maximized over all incoherent operations $\Lambda$. Using Eq.~(\ref{eq:MCoverlap}) in Eq.~(\ref{eq:MCspeedlimitProof}) completes the proof.
\end{proof}

\medskip{}

\textbf{\emph{Conclusions and outlook.}} We have investigated speed limits for basis change via unitary evolutions, providing bounds on the evolution time which are optimal for several interesting scenarios. 

For dimensions $d \leq 4$ we found the optimal evolution time required to convert the computational basis into an unbiased, i.e., maximally coherent basis. Perhaps surprisingly, the minimal evolution times coincide for $d=2$ and $d=4$, when Hamiltonians with the same mean energy $E$ are considered. Moreover, for $d=3$ the saturation of the speed limit prefers a special ordering of the basis that is unbiased with respect to the computational basis. We also showed that an $n$-qubit Hadamard gate can be implemented within time $\pi/2E$. This proves that in multi-qubit systems, a maximally coherent basis can be established within a period of time which is independent on the number of qubits. These results further imply that in multi-qubit systems interactive Hamiltonians can significantly reduce the evolution time, compared to the time for establishing an unbiased basis by evolving each qubit independently. We further showed that in the limit $d \rightarrow \infty$ the time for establishing an unbiased basis is at least $\pi/4E$. Speed limits for basis permutation are also discussed.

We have also investigates speed limits for generating a certain amount of quantum coherence, as well as minimal time to convert a pure state into a maximally coherent one.  We expect that our methods can also be used to derive minimal transformation times for general bases and other quantum resources, such as quantum entanglement and imaginarity~\cite{Hickey_2018,WuPhysRevLett.126.090401,WuPhysRevA.103.032401}.

\textbf{\emph{Acknowledgements.}} This work was supported by the National Science Centre, Poland, within the QuantERA II Programme (No 2021/03/Y/ST2/00178, acronym ExTRaQT) that has received funding from the European Union's Horizon 2020 research and innovation programme under Grant Agreement No 101017733 and the ``Quantum Coherence and Entanglement for Quantum Technology'' project, carried out within the First Team programme of the Foundation for Polish Science co-financed by the European Union under the European Regional Development Fund. P.H. acknowledges support by the Foundation for Polish Science
(IRAP project, ICTQT, contract no. 2018/MAB/5,
co-financed by EU within Smart Growth Operational
Programme). C.M. and D.B. acknowledge support by the EU QuantERA project QuICHE.

\bibliography{literature}


\appendix


\section{\label{sec:SingleQubit}Speed limits for single-qubit states}

A general single-qubit Hamiltonian has the form 
\begin{equation}
H=E_{+}\ket{E_{+}}\!\bra{E_{+}}+E_{-}\ket{E_{-}}\!\bra{E_{-}},\label{eq:H-Qubit-1}
\end{equation}
where the eigenvalues $E_{\pm}$ and eigenstates $\ket{E_{\pm}}$
can be parametrized as 
\begin{equation}
E_{\pm} = \left(G\pm E\right),\,\,\,\,\,\ket{E_{\pm}}\!\bra{E_{\pm}}=\frac{1}{2}\left(\openone\pm\boldsymbol{n}\cdot\boldsymbol{\sigma}\right).
\end{equation}
Here, $G$ and $E\geq0$ are real numbers, $\boldsymbol{n}=(n_{x},n_{y},n_{z})$
is a normalized vector, and $\boldsymbol{\sigma}=(\sigma_{x},\sigma_{y},\sigma_{z})$
contains the three Pauli operators. The Hamiltonian~(\ref{eq:H-Qubit-1})
can thus be equivalently expressed as 
\begin{equation}
H = \left(E\boldsymbol{n}\cdot\boldsymbol{\sigma}+G\openone\right).\label{eq:H-Qubit-2}
\end{equation}
Note that $E$ corresponds to the mean energy of the Hamiltonian:
\begin{equation}
E=\frac{1}{2}\mathrm{Tr}[H]-E_{-}.
\end{equation}
Equipped with these tools, we will now present a bound for the evolution
time between any two single-qubit states. 
\begin{prop}
\label{prop:tmin}The time for converting a single-qubit state $\rho_{0}$
into the state $\rho_{1}$ via unitary evolution $U=e^{-iHt}$ is
bounded as 
\begin{equation}
T(\rho_{0}\rightarrow\rho_{1})\geq\frac{1}{2E}\arccos\left(\frac{\boldsymbol{r}_{0}\cdot\boldsymbol{r}_{1}}{|\boldsymbol{r}_{0}||\boldsymbol{r}_{1}|}\right),\label{eq:tmin}
\end{equation}
where $\boldsymbol{r}_{i}$ is the Bloch vector of the state $\rho_{i}$. 
\end{prop}

\begin{proof}
Note that the unitary 
\begin{equation}
U(t)=e^{-iHt}=e^{-iGt}e^{-iEt\boldsymbol{n}\cdot\boldsymbol{\sigma}}
\end{equation}
can be interpreted as a rotation by an angle $2Et$ about the axis
$\boldsymbol{n}$ of the Bloch sphere. The minimal value for $Et$
is achieved by choosing the rotation axis $\boldsymbol{n}$ to be
orthogonal to both Bloch vectors $\boldsymbol{r}_{0}$ and $\boldsymbol{r}_{1}$:
\begin{align}
\boldsymbol{n} & =\frac{\boldsymbol{r}_{0}\times\boldsymbol{r}_{1}}{\left|\boldsymbol{r}_{0}\times\boldsymbol{r}_{1}\right|},\\
Et & = \frac{1}{2}\arccos\left(\frac{\boldsymbol{r}_{0}\cdot\boldsymbol{r}_{1}}{|\boldsymbol{r}_{0}||\boldsymbol{r}_{1}|}\right).\label{eq:tmin-2}
\end{align}
This completes the proof of the proposition. 
\end{proof}
Noting that $\mathrm{Tr}[\rho_{i}\rho_{j}]=(1+\boldsymbol{r}_{i}\cdot\boldsymbol{r}_{j})/2$
we can reformulate Eq.~(\ref{eq:tmin}) as follows: 
\begin{equation}
T(\rho_{0}\rightarrow\rho_{1})\geq\frac{1}{2E}\arccos\left(\frac{2\mathrm{Tr}[\rho_{0}\rho_{1}]-1}{\sqrt{(2\mathrm{Tr}[\rho_{0}^{2}]-1)(2\mathrm{Tr}[\rho_{1}^{2}]-1)}}\right).\label{eq:BoundQubit}
\end{equation}
The proof of Proposition~\ref{prop:tmin} implies that this bound
is tight, i.e., for any two single qubit-states $\rho_{0}$ and $\rho_{1}$,
there exists a Hamiltonian with mean energy $E$ saturating Eq.~(\ref{eq:BoundQubit}).
For pure qubit states this expression simplifies to the tight bound 
\begin{equation}
T(\ket{\psi_{0}}\rightarrow\ket{\psi_{1}})\geq\frac{1}{2E}\arccos\left(2|\!\braket{\psi_{0}|\psi_{1}}\!|^{2}-1\right).\label{eq:pure}
\end{equation}

For single-qubit systems, any unitary transforming $\ket{0}$ into $\ket{+} = (\ket{0}+\ket{1})/\sqrt{2}$ also transforms $\ket{1}$ into $\ket{-} = (\ket{0}-\ket{1})/\sqrt{2}$. For a transition from the computational basis $\{\ket{0},\ket{1}\}$
to an unbiased unbiased qubit basis we thus obtain
\begin{align}
T_{\mathrm{unbiased}} & \geq\frac{\pi}{4E},
\end{align}
as claimed in the main text.

\section{\label{sec:Unbiased-Qutrits}Unbiased bases for qutrits}

Up to an overall phase for each basis element, an arbitrary unbiased
basis (w.r.t. the computational basis) for a qutrit can be written
as \begin{subequations} 
\begin{align}
\ket{0_{+}} & =\frac{1}{\sqrt{3}}(\ket{0}+e^{i\alpha_{0,1}}\ket{1}+e^{i\alpha_{0,2}}\ket{2}),\\
\ket{1_{+}} & =\frac{1}{\sqrt{3}}(\ket{0}+e^{i\alpha_{1,1}}\ket{1}+e^{i\alpha_{1,2}}\ket{2}),\\
\ket{2_{+}} & =\frac{1}{\sqrt{3}}(\ket{0}+e^{i\alpha_{2,1}}\ket{1}+e^{i\alpha_{2,2}}\ket{2}),
\end{align}
\end{subequations} where the phases $\alpha_{i,j}$ need to fulfill
the condition 
\begin{equation}
1+e^{i(\alpha_{k,1}-\alpha_{l,1})}+e^{i(\alpha_{k,2}-\alpha_{l,2})}=3\delta_{k,l}.
\end{equation}
This condition determines the form of the basis to be either

\begin{subequations} \label{eq:UnbiasedQutrit1-1} 
\begin{align}
\ket{0_{+}} & =\frac{1}{\sqrt{3}}\left(\ket{0}+e^{i\left(\alpha_{0,1}+\frac{2}{3}\pi\right)}\ket{1}+e^{i\left(\alpha_{0,2}+\frac{4}{3}\pi\right)}\ket{2}\right),\\
\ket{1_{+}} & =\frac{1}{\sqrt{3}}\left(\ket{0}+e^{i\alpha_{0,1}}\ket{1}+e^{i\alpha_{0,2}}\ket{2}\right),\\
\ket{2_{+}} & =\frac{1}{\sqrt{3}}\left(\ket{0}+e^{i\left(\alpha_{0,1}-\frac{2}{3}\pi\right)}\ket{1}+e^{i\left(\alpha_{0,2}-\frac{4}{3}\pi\right)}\ket{2}\right),
\end{align}
\end{subequations} or \begin{subequations} \label{eq:UnbiasedQutrit2-1}
\begin{align}
\ket{0_{+}} & =\frac{1}{\sqrt{3}}\left(\ket{0}+e^{i\left(\alpha_{0,1}-\frac{2}{3}\pi\right)}\ket{1}+e^{i\left(\alpha_{0,2}-\frac{4}{3}\pi\right)}\ket{2}\right),\\
\ket{1_{+}} & =\frac{1}{\sqrt{3}}\left(\ket{0}+e^{i\alpha_{0,1}}\ket{1}+e^{i\alpha_{0,2}}\ket{2}\right),\\
\ket{2_{+}} & =\frac{1}{\sqrt{3}}\left(\ket{0}+e^{i\left(\alpha_{0,1}+\frac{2}{3}\pi\right)}\ket{1}+e^{i\left(\alpha_{0,2}+\frac{4}{3}\pi\right)}\ket{2}\right).
\end{align}
\end{subequations} If we now introduce the unbiased bases \begin{subequations}\label{eq:Unbiased-1-1}
\begin{align}
\ket{0_{+}} & =\frac{1}{\sqrt{3}}\left(\ket{0}+e^{i\frac{2}{3}\pi}\ket{1}+e^{i\frac{4}{3}\pi}\ket{2}\right),\\
\ket{1_{+}} & =\frac{1}{\sqrt{3}}\left(\ket{0}+\ket{1}+\ket{2}\right),\\
\ket{2_{+}} & =\frac{1}{\sqrt{3}}\left(\ket{0}+e^{-i\frac{2}{3}\pi}\ket{1}+e^{-i\frac{4}{3}\pi}\ket{2}\right),
\end{align}
\end{subequations} and \begin{subequations}\label{eq:Unbiased-2-1}
\begin{align}
\ket{\tilde{0}_{+}} & =\frac{1}{\sqrt{3}}\left(\ket{0}+e^{-i\frac{2}{3}\pi}\ket{1}+e^{-i\frac{4}{3}\pi}\ket{2}\right),\\
\ket{\tilde{1}_{+}} & =\frac{1}{\sqrt{3}}\left(\ket{0}+\ket{1}+\ket{2}\right),\\
\ket{\tilde{2}_{+}} & =\frac{1}{\sqrt{3}}\left(\ket{0}+e^{i\frac{2}{3}\pi}\ket{1}+e^{i\frac{4}{3}\pi}\ket{2}\right),
\end{align}
\end{subequations} we see that any basis of the form~(\ref{eq:UnbiasedQutrit1-1})
or (\ref{eq:UnbiasedQutrit2-1}) can be obtained from the basis~(\ref{eq:Unbiased-1-1})
or (\ref{eq:Unbiased-2-1}), respectively, by using the diagonal unitary
$V=\ket{0}\!\bra{0}+e^{i\alpha_{0,1}}\ket{1}\!\bra{1}+e^{i\alpha_{0,2}}\ket{2}\!\bra{2}$.

\section{\label{sec:ProofV}Speed limits for unitary rotated bases}

Let $\{\ket{\psi_j}\}$ and $\{\ket{\phi_j}\}$ be two complete orthonormal bases. A speed limit of the form 
\begin{equation}
T\left(\ket{\psi_{j}}\rightarrow\ket{\phi_{j}}\right)\geq\frac{g}{E}\label{eq:SpeedLimitUnitary}
\end{equation}
directly leads to a speed limit for any basis which can be obtained
from $\{\ket{\phi_{j}}\}$ via a unitary $V=\sum_{j}e^{i\alpha_{j}}\ket{\psi_{j}}\!\bra{\psi_{j}}$:
\begin{equation}
T\left(\ket{\psi_{j}}\rightarrow V\ket{\phi_{j}}\right)\geq\frac{g}{E}.\label{eq:GeneralSpeedLimitV}
\end{equation}
The speed limit~(\ref{eq:GeneralSpeedLimitV}) is tight whenever
Eq.~(\ref{eq:SpeedLimitUnitary}) is tight. To prove this, let $H$
be a Hamiltonian such that 
\begin{equation}
e^{-iHt}\ket{\psi_{j}}=\ket{\phi_{j}}.
\end{equation}
Then the Hamiltonian $H'=VHV^{\dagger}$ achieves the transformation
\begin{equation}
e^{-iH't}\ket{\psi_{j}}=e^{-i\alpha_{j}}V\ket{\phi_{j}},
\end{equation}
which can be seen by using the expression $e^{-iH't}=Ve^{-iHt}V^{\dagger}$.
Noting that $H$ and $H'$ have the same mean energy $E$, we see
that Eq.~(\ref{eq:SpeedLimitUnitary}) implies the speed limit~(\ref{eq:GeneralSpeedLimitV})
for any unitary $V$ which is diagonal in the $\{\ket{\psi_{j}}\}$
basis. Moreover, the speed limit~(\ref{eq:GeneralSpeedLimitV}) is
tight for all diagonal unitaries $V$ whenever Eq.~(\ref{eq:SpeedLimitUnitary})
is tight.

As we have seen in Appendix~\ref{sec:Unbiased-Qutrits}, any unbiased
basis of a qutrit can be created from the basis $\{\ket{n_{+}}\}$
or $\{\ket{\tilde{n}_{+}}\}$ {[}see Eqs.~(\ref{eq:Unbiased-1-1})
and~(\ref{eq:Unbiased-2-1}){]} via a diagonal unitary $V$. In combination
with the arguments mentioned above, this implies that speed limits
for the transitions $\{\ket{n}\}\rightarrow\{\ket{n_{+}}\}$ and $\{\ket{n}\}\rightarrow\{\ket{\tilde{n}_{+}}\}$
will also lead to speed limits for general unbiased qutrit bases $\{\ket{n}\}\rightarrow\{V\ket{n_{+}}\}$
and $\{\ket{n}\}\rightarrow\{V\ket{\tilde{n}_{+}}\}$.

\section{\label{sec:Proof-Theorem1}Proof of Theorem~\ref{thm:UnbiasedQutrit}}
Before we focus on the case $d=3$ we will discuss the problem for general $d$. For this, let $U=e^{-iHt}$ be a unitary achieving
the transformation $\{\ket{n}\} \rightarrow \{\ket{n_{+}}\}$, where $\{\ket{n_{+}}\}$ is now a maximally coherent basis of dimension $d$. Any unitary achieving the desired transformation
must be of the form 
\begin{equation}
U=\sum_{n=0}^{d-1}e^{i\phi_{n}}\ket{n_{+}}\!\bra{n}
\end{equation}
 with some phases $\phi_{n}$. We further obtain 
\begin{equation}
\mathrm{Tr}[U+U^{\dagger}]=\sum_{n=0}^{d-1}\left(e^{i\phi_{n}}\braket{n|n_{+}}+e^{-i\phi_{n}}\braket{n_{+}|n}\right).
\end{equation}
Noting that $\braket{n|n_{+}}=e^{i\gamma_{n}}/\sqrt{d}$ with some
phases $\gamma_{n}$ we arrive at the inequality
\begin{equation}
-2\sqrt{d}\leq\mathrm{Tr}[U+U^{\dagger}]\leq2\sqrt{d}.
\end{equation}
On the other hand, recalling that $U=e^{-iHt}$ with a Hamiltonian
$H$ we obtain 
\begin{equation}
\mathrm{Tr}[U+U^{\dagger}]=2\sum_{i}\cos(E_{i}t),
\end{equation}
where $E_{i}$ are the eigenvalues of the Hamiltonian. In summary, for any unitary transformation $U=e^{-iHt}$ leading
to the transformation $\{\ket{n}\}\rightarrow\{\ket{n_{+}}\}$ it
must hold that 
\begin{equation}
-\sqrt{d}\leq\sum_{i}\cos(E_{i}t)\leq\sqrt{d}.\label{eq:Constraint}
\end{equation}

We will now consider $d=3$. In this case, we will show that any unitary $U=e^{-iHt}$ leading to the transformation
$\{\ket{n}\}\rightarrow\{\ket{n_{+}}\}$ fulfills 
\begin{equation}
Et\geq\frac{2}{9}\pi.\label{eq:QutritBound2}
\end{equation}
Assuming that $E_{i}$ are in increasing order, we see that $E\geq (E_{2}-E_{0})/3$.
Thus, for proving Eq.~(\ref{eq:QutritBound2}) it is enough to prove
that 
\begin{equation}
(E_{2}-E_{0})t\geq\frac{2}{3}\pi.\label{eq:QutritBound3}
\end{equation}
We will prove this by contradiction, assuming that the transformation
is possible with a unitary violating Eq.~(\ref{eq:QutritBound3}).
Violation of Eq.~(\ref{eq:QutritBound3}) implies that 
\begin{equation}
(E_{1}-E_{0})t\leq\frac{\pi}{3}\,\,\,\,\mathrm{or}\,\,\,\,(E_{2}-E_{1})t\leq\frac{\pi}{3}.
\end{equation}
In the first case $(E_{1}-E_{0})t\leq\pi/3$, we can set (without
loss of generality) $E_{0}t=-\pi/6$, which implies the inequalities
\begin{equation}
\left|E_{1}t\right|\leq\frac{\pi}{6},\,\,\,\,E_{2}t<\frac{\pi}{2}.
\end{equation}
It follows that 
\begin{equation}
\sum_{i}\cos(E_{i}t)>2\cos\left(\frac{\pi}{6}\right),\label{eq:QutritContradiction}
\end{equation}
which is a contradiction to Eq.~(\ref{eq:Constraint}). The remaining
case $(E_{2}-E_{1})t\leq\pi/3$ can be treated similarly, by choosing
(without loss of generality) $E_{2}t=\pi/6$, thus obtaining the following
inequalities: 
\begin{equation}
\left|E_{1}t\right|\leq\frac{\pi}{6},\,\,\,\,E_{0}t>-\frac{\pi}{2}.
\end{equation}
Also in this case we obtain the inequality~(\ref{eq:QutritContradiction}),
in contradiction to Eq.~(\ref{eq:Constraint}). This completes
the proof of the bound~(\ref{eq:QutritBound2}). Since the methods presented above apply for any qutrit basis which is unbiased with respect to the computational basis, this completes the proof of Theorem~\ref{thm:UnbiasedQutrit}.

\section{\label{sec:Proof-Proposition2}Proof of Proposition~\ref{prop:tight}}

According to Theorem~\ref{thm:UnbiasedQutrit}, we have the following inequalities for transition into the bases~(\ref{eq:Unbiased-1-1})
and~(\ref{eq:Unbiased-2-1}):
\begin{subequations} 
\begin{align}
T(\ket{n} & \rightarrow\ket{n_{+}})\geq\frac{2\pi}{9E},\label{eq:UnbiasedQutritBoundAppendix-1}\\
T(\ket{n} & \rightarrow\ket{\tilde{n}_{+}})\geq\frac{2\pi}{9E}. \label{eq:eq:UnbiasedQutritBoundAppendix-2}
\end{align}
\end{subequations}  
As can be checked by inspection, Eq.~(\ref{eq:UnbiasedQutritBoundAppendix-1})
is saturated for the basis~(\ref{eq:Unbiased-1-1}) by the Hamiltonian
$H=\ket{\alpha}\!\bra{\alpha}$ with 
\begin{equation}
\ket{\alpha}=\frac{1}{\sqrt{3}}\left(\ket{0}+e^{-i\frac{2}{3}\pi}\ket{1}+\ket{2}\right).\label{eq:alpha-2}
\end{equation}

We will now prove that the inequality~(\ref{eq:eq:UnbiasedQutritBoundAppendix-2})
is strict for the basis~(\ref{eq:Unbiased-2-1}), i.e., there is
no evolution $e^{-iHt}$ leading to the transformation $\ket{n}\rightarrow\ket{\tilde{n}_{+}}$
within the time $t=2\pi/(9E)$. Assume -- by contradiction -- that
the bound is saturated for some unitary $U=e^{-iHt}$: 
\begin{equation}
\ket{\tilde{n}_{+}}=e^{-iHt}\ket{n},\,\,\,\,t=\frac{2\pi}{9E}.\label{eq:Contradiction}
\end{equation}
Recalling that $E_{i}$ are in decreasing order and following the
arguments from the proof of Theorem~\ref{thm:UnbiasedQutrit}, it
must be that 
\begin{align}
E_1 &= E_0, \\
(E_{2}-E_{0})t &= \frac{2}{3}\pi.
\end{align}
Without loss of generality we can choose 
\begin{equation}
E_{0}t=E_{1}t=-\frac{\pi}{6},\,\,\,\,E_{2}t=\frac{\pi}{2}.
\end{equation}
Summarizing these arguments, there exists a unitary $U=e^{-iHt}$
fulfilling Eq.~(\ref{eq:Contradiction}) and having eigenvalues 
\begin{equation}
\lambda_{0}=\lambda_{1}=e^{i\frac{\pi}{6}},\,\,\,\,\lambda_{2}=e^{-i\frac{\pi}{2}},\label{eq:lambda}
\end{equation}
which implies that it fulfills 
\begin{equation}
\mathrm{Tr}[U+U^{\dagger}]=2\sqrt{3}.\label{eq:Extreme}
\end{equation}

On the other hand, the unitary also admits the form
\begin{equation}
    U=\sum_{n=0}^{2}e^{i\phi_{n}}\ket{\tilde{n}_{+}}\!\bra{n},
\end{equation}
with some phases $\phi_n$. We find that 
\begin{align}
\mathrm{Tr}\left[U+U^{\dagger}\right] & =\frac{2}{\sqrt{3}}\left(\cos\phi_{0}+\cos\phi_{1}\right)\nonumber \\
 & -\frac{1}{\sqrt{3}}\cos\phi_{2}+\sin\phi_{2}.
\end{align}
Together with Eq.~(\ref{eq:Extreme}) we obtain
\begin{align}
\frac{2}{\sqrt{3}}\left(\cos\phi_{0}+\cos\phi_{1}\right) -\frac{1}{\sqrt{3}}\cos\phi_{2}+\sin\phi_{2} = 2 \sqrt{3}.
\end{align}
This equation has a unique solution in the range $0\leq\phi_{i}\leq2\pi$, given by
\begin{equation}
\phi_{0}=\phi_{1}=0,\,\,\,\,\phi_{2}=\frac{2}{3}\pi.\label{eq:QutritSolution-2}
\end{equation}
This implies that the eigenvalues of $U$ must be 
\begin{equation}
\mu_{0}=\mu_{1}=e^{-i\frac{\pi}{6}},\,\,\,\,\mu_{2}=e^{i\frac{\pi}{2}},
\end{equation}
which is a contradiction to Eq.~(\ref{eq:lambda}). This completes
the proof of the proposition.

\section{Proof of Theorem~\ref{thm:UnbiasedTwoQubits}} \label{sec:ProofUnbiasedTwoQubits}
We will now focus on the case $d=4$. For this case we will prove the lower bound 
\begin{equation}
    Et \geq \frac{\pi}{4}. \label{eq:ProofD4bound}
\end{equation}
We will prove this by contradiction, assuming that there exists a unitary $U=e^{-iHt}$ transforming $\{\ket{n}\}$ onto a maximally coherent basis with 
\begin{equation}
    Et < \frac{\pi}{4}. \label{eq:ProofD4}
\end{equation}
Without loss of generality we can assume that $E_0=0$, which implies $E=(E_1+E_2+E_3)/4$.

We now define $\alpha_i = E_i t$. Note that $\pi > \alpha_i \geq 0$. Due to Eq.~(\ref{eq:ProofD4}) we have $\alpha_3 < \pi - \alpha_1 - \alpha_2$, which further implies
\begin{equation}
    \cos (\alpha_3) > \cos (\pi - \alpha_1 - \alpha_2) = - \cos (\alpha_1 + \alpha_2).
\end{equation}
It follows that
\begin{align}
    \cos(\alpha_1) + \cos(\alpha_2) + \cos(\alpha_3) &> \cos(\alpha_1) + \cos(\alpha_2) \label{eq:ProofD5} \\
    &- \cos(\alpha_1 + \alpha_2). \nonumber
\end{align}
We will now investigate closer the right-hand side of Eq.~(\ref{eq:ProofD5}), defining 
\begin{equation}
    f(\boldsymbol \alpha) = \cos(\alpha_1) + \cos(\alpha_2) - \cos(\alpha_1 + \alpha_2).
\end{equation}
In particular, we will show that $f(\boldsymbol{\alpha}) \geq 1$ holds true whenever \begin{subequations} \label{eq:ProofD8}
\begin{align}
    \alpha_i &\geq 0, \\
    \alpha_1 + \alpha_2 &\leq \pi.
\end{align}
\end{subequations}
For this, we evaluate the partial derivatives of $f$ with respect to $\alpha_i$:
\begin{align}
    \frac{\partial f}{\partial \alpha_1} = \sin (\alpha_1 + \alpha_2) - \sin (\alpha_1), \\
    \frac{\partial f}{\partial \alpha_2} = \sin (\alpha_1 + \alpha_2) - \sin (\alpha_2).
\end{align}
To find local extrema of $f$ we set $\partial f / \partial \alpha_i = 0$, which implies $\sin(\alpha_1) = \sin (\alpha_2)$. This means that $\alpha_1 = \alpha_2$, or $\alpha_1 = \pi - \alpha_2$. With the condition $\alpha_1 = \alpha_2$ we further obtain $\sin(2\alpha_2) = \sin(\alpha_2)$, with the solutions
\begin{subequations} \label{eq:ProofD6}
\begin{align}
    \alpha_1 = \alpha_2 = 0,\\
    \alpha_1 = \alpha_2 = \frac{\pi}{3}.
\end{align}
\end{subequations}
On the other hand, the condition $\alpha_1 = \pi - \alpha_2$ together with $\partial f / \partial \alpha_i = 0$ leads to $\sin(\alpha_1) = \sin(\alpha_2) = 0$, with the solutions 
\begin{subequations} \label{eq:ProofD7}
\begin{align}
    \alpha_1 &= 0, \alpha_2 = \pi, \\
    \alpha_1 &= \pi, \alpha_2 = 0.
\end{align}
\end{subequations}
For proving that $f(\boldsymbol{\alpha}) \geq 1$ we evaluate $f(\boldsymbol{\alpha})$ at the extrema (\ref{eq:ProofD6}) and (\ref{eq:ProofD7}), and also at the boundary of the region defined in Eqs.~(\ref{eq:ProofD8}). For the solutions~(\ref{eq:ProofD6}) we obtain $f(\boldsymbol{\alpha})=1$ and $f(\boldsymbol{\alpha})=3/2$, respectively. Moreover, the solutions (\ref{eq:ProofD7}) give $f(\boldsymbol \alpha) = 1$. 

It remains to show that $f(\boldsymbol\alpha) \geq 1$ also at the boundary of the region defined in Eqs.~(\ref{eq:ProofD8}). For a given value of $\alpha_1 \in [0,\pi]$, the boundary is attained for $\alpha_2 = 0$ or $\alpha_2 = \pi-\alpha_1$. As one can verify by inspection, $f(\boldsymbol{\alpha})=1$ in both cases. In summary, this proves that $f(\boldsymbol \alpha) \geq 1$ within the region (\ref{eq:ProofD8}). 

Collecting the above arguments, Eq.~(\ref{eq:ProofD4}) implies that there is a unitary $U=e^{-iHt}$ achieving the transformation $\{\ket{n}\}\rightarrow\{\ket{n_{+}}\}$ with $\sum_i \cos (E_i t) > 2$, in contradiction to Eq.~(\ref{eq:Constraint}). This completes the proof of the lower bound~(\ref{eq:ProofD4bound}).

As is explained in the main text, it is indeed possible to achieve the transformation $\{\ket{n}\}\rightarrow\{\ket{n_{+}}\}$ within time $t = \pi/(4 E)$. This completes the proof of the theorem.

\section{\label{sec:Permutation}Eigenvalues of permutation unitary}

In the following we will determine the eigenvalues of the permutation
unitary 
\begin{equation}
U\ket{n}=\ket{(n+1)\!\!\!\!\!\mod\!d}.\label{eq:Uperm}
\end{equation}
Let $\ket{\psi}=\sum_{n}a_{n}\ket{n}$ be an eigenstate of $U$, i.e.,
\begin{equation}
U\ket{\psi}=e^{i\alpha}\ket{\psi}.
\end{equation}
From Eq.~(\ref{eq:Uperm}) we obtain 
\begin{equation}
a_{n}=e^{i\alpha}a_{(n+1)\!\!\!\!\!\mod\!d},
\end{equation}
which implies that all coefficients $a_{n}$ must have the same absolute
value: $|a_{n}|^{2}=1/d$. Thus, any eigenstate $\ket{\psi}$ has
the form 
\begin{equation}
\ket{\psi}=\frac{1}{\sqrt{d}}\sum_{j=0}^{d-1}e^{i\phi_{j}}\ket{j}.\label{eq:PermuteEigenstate}
\end{equation}
From this it follows that $U$ cannot be degenerate. To prove this,
assume -- by contradiction -- that there exists two eigenstates
$\ket{\psi_{1}}$ and $\ket{\psi_{2}}$ with the same eigenvalue.
Then, any superposition of $\ket{\psi_{1}}$ and $\ket{\psi_{2}}$
is also an eigenstate of $U$. Moreover, by superposing $\ket{\psi_{1}}$
and $\ket{\psi_{2}}$ we can obtain an eigenstate which is not of
the form~(\ref{eq:PermuteEigenstate}), which is the desired contradiction.

In the next step note that any permutation unitary must fulfill 
\begin{equation}
U^{d}=\openone.
\end{equation}
Together with the fact that $U$ is non-degenerate, the eigenvalues
of $U$ must be of the form 
\[
\lambda_{n}=e^{i\frac{2\pi}{d}n},
\]
where $n$ is an integer in the range $0\leq n\leq d-1$.

\section{\label{sec:PureStates}Speed limits for pure states}

Let $H$ be a Hamiltonian of dimension $d$ with eigenvalues $E_{i}$
and eigenstates $\ket{E_{i}}$. Without loss of generality, we assume
that the eigenvalues are in increasing order, and thus $E_{\max}=E_{d-1}$
and $E_{\min}=E_{0}$.

Suppose now that an initial state $\ket{\psi}$ evolves for the time
$0\leq t\leq\pi/E_{\mathrm{gap}}$, where $E_{\mathrm{gap}}=E_{\max}-E_{\min}$
is the energy gap of the Hamiltonian. In the following, we are interested
in the minimal overlap between the initial state $\ket{\psi}$ and
the time-evolved state $\ket{\psi_{t}}=e^{-iHt}\ket{\psi}$: 
\begin{equation}
F_{\min}=\min_{\ket{\psi}}\left|\braket{\psi|e^{-iHt}|\psi}\right|,
\end{equation}
minimized over all initial states $\ket{\psi}$. 
\begin{prop}
\label{prop:Fmin}For a given Hamiltonian $H$ and evolution time
$0\leq t\leq\pi/E_{\mathrm{gap}}$ it holds that 
\begin{align}
F_{\min} & =\left|\braket{\psi_{\min}|e^{-iHt}|\psi_{\min}}\right|=\frac{1}{2}\left|e^{-iE_{\mathrm{gap}}t}+1\right|
\end{align}
with $\ket{\psi_{\min}}=\frac{1}{\sqrt{2}}\left(\ket{E_{0}}+\ket{E_{d-1}}\right)$. 
\end{prop}

\begin{proof}
Expanding the initial state in the eigenbasis of the Hamiltonian as
$\ket{\psi}=\sum_{j}c_{j}\ket{E_{j}}$ with complex coefficients $c_{j}$
allows us to write the overlap $|\!\braket{\psi|e^{-iHt}|\psi}\!|$
as follows: 
\begin{equation}
\left|\braket{\psi|e^{-iHt}|\psi}\right|=\left|\sum_{j}|c_{j}|^{2}e^{-iE_{j}t}\right|.
\end{equation}
Noting that the coefficients $c_{j}$ fulfill the condition $\sum_{j}|c_{j}|^{2}=1$,
our figure of merit can be expressed as 
\begin{equation}
F_{\min}=\min_{\ket{\psi}}\left|\braket{\psi|e^{-iHt}|\psi}\right|=\min_{\{p_{j}\}}\left|\sum_{j}p_{j}e^{-iE_{j}t}\right|,
\end{equation}
where the minimum on the right-hand side is taken over all probability
distributions $\{p_{j}\}$. Recalling that $E_{\mathrm{gap}}t\leq\pi$,
it is straightforward to see that the minimum is attained for the
following choice of $\{p_{j}\}$: 
\begin{equation}
p_{j}=\begin{cases}
\frac{1}{2} & \mathrm{for}\,\,\,j=0\,\,\mathrm{and}\,\,j=d-1,\\
0 & \mathrm{for}\,\,\,0<j<d-1.
\end{cases}
\end{equation}
It follows that the optimal state $\ket{\psi_{\min}}$, minimizing
the overlap $|\!\braket{\psi|e^{-iHt}|\psi}\!|$, can be chosen as
\begin{equation}
\ket{\psi_{\min}}=\frac{1}{\sqrt{2}}\left(\ket{E_{0}}+\ket{E_{d-1}}\right),
\end{equation}
as claimed. In the last step, it is straightforward to verify that
\begin{equation}
\left|\braket{\psi_{\min}|e^{-iHt}|\psi_{\min}}\right|=\frac{1}{2}\left|e^{-iE_{\mathrm{gap}}t}+1\right|
\end{equation}
which completes the proof of the proposition. 
\end{proof}
Remarkably, $F_{\min}$ does not depend on the structure of the Hamiltonian,
but only on the gap between the largest and the smallest eigenvalue
$E_{\mathrm{gap}}$. In the following, we will use this result to
bound the evolution time between pure states. 
\begin{prop}
\label{prop:Pure}The time for converting a pure states $\ket{\psi_{0}}$
into another state $\ket{\psi_{1}}$ via unitary evolution $U=e^{-iHt}$
is bounded as 
\begin{equation}
T(\ket{\psi_{0}}\rightarrow\ket{\psi_{1}})\geq\frac{1}{E_{\mathrm{gap}}}\arccos\left(2|\!\braket{\psi_{0}|\psi_{1}}\!|^{2}-1\right).\label{eq:tmin-3}
\end{equation}
\end{prop}

\begin{proof}
If the states $\ket{\psi_{0}}$ and $\ket{\psi_{1}}$ fulfill $\ket{\psi_{1}}=e^{-iHt}\ket{\psi_{0}}$
with $0\leq t\leq\pi/E_{\mathrm{gap}}$, then by Proposition~\ref{prop:Fmin}
it follows that 
\begin{equation}
\left|\braket{\psi_{0}|\psi_{1}}\right|^{2}\geq\frac{1}{4}\left|e^{-iE_{\mathrm{gap}}t}+1\right|^{2}.
\end{equation}
This inequality is equivalent to 
\begin{equation}
t\geq\frac{1}{E_{\mathrm{gap}}}\arccos\left(2\left|\braket{\psi_{0}|\psi_{1}}\right|^{2}-1\right).
\end{equation}

On the other hand, if $\ket{\psi_{0}}$ and $\ket{\psi_{1}}$ fulfill $\ket{\psi_{1}}=e^{-iHt}\ket{\psi_{0}}$ with $t > \pi/E_{\mathrm{gap}}$, Eq.~(\ref{eq:tmin-3}) is automatically satisfied, since $\arccos(x) \leq \pi/2$ for $x \geq 0$. This completes the proof.

\end{proof}

Noting that $E_{\mathrm{gap}}\leq dE$, where $E=\mathrm{Tr}[H]/d - E_{0}$
is the average energy of the Hamiltonian, we immediately obtain the following lemma.
\begin{lem} \label{lem:PureStates}
The time for converting a pure state $\ket{\psi_0}$ into another state $\ket{\psi_1}$ via unitary evolution $U = e^{-iHt}$ is bounded below as
\begin{equation}
T(\ket{\psi_{0}}\rightarrow\ket{\psi_{1}})\geq\frac{1}{dE}\arccos\left(2|\!\braket{\psi_{0}|\psi_{1}}\!|^{2}-1\right).\label{eq:BoundPure}
\end{equation}
\end{lem}

Moreover, for any two pure states $\ket{\psi_{0}}$ and $\ket{\psi_{1}}$
there exists a Hamiltonian $H$ saturating Eq.~(\ref{eq:BoundPure}).
To see this, recall that Eq.~(\ref{eq:BoundPure}) is tight for $d=2$, see also Eq.~(\ref{eq:pure}). Let now $H = \ket{\phi}\!\bra{\phi}$ be a Hamiltonian which saturates the inequality for $d=2$. Note that the mean energy in this case is given by $E=1/2$. This implies that the Hamiltonian achieves the transformation $\ket{\psi_0} \rightarrow \ket{\psi_1}$ within the time 
\begin{equation}
    t = \arccos\left(2|\!\braket{\psi_{0}|\psi_{1}}\!|^{2}-1\right), \label{eq:PureShortestTime}
\end{equation}
which is the shortest possible time for $E=1/2$. For $d > 2$ we can use the same Hamiltonian $H = \ket{\phi}\!\bra{\phi}$ to achieve the transformation within the same time as given in Eq.~(\ref{eq:PureShortestTime}). The mean energy is now given by $E = 1/d$, and we see that Eq.~(\ref{eq:BoundPure}) is saturated.

\section{Proof of Theorem \ref{thm:LowerBoundGeneral}} \label{sec:ProofLowerBound}

We define $T_{low} = \frac{d-1}{dE}\frac{\pi}{4}$ and $d\geq 2$. Let us assume that $T\leq T_{low}$. Then there must exist a Hamiltonian such that:
\begin{equation}
    ET\leq \frac{d-1}{d}\frac{\pi}{4}.
\end{equation}
Without loss of generality, we consider $E_{0}=0$ and $E_{j}\geq 0$ for all $j$. Also we define $\alpha_{j}=E_{j}T$, therefore we have:
\begin{equation}
    \sum_{j}\alpha_{j}\leq (d-1)\frac{\pi}{4}.
    \label{Region}
\end{equation}

By Eq. (\ref{eq:Constraint}) we must have $-\sqrt{d} \leq \sum_{j}\cos{\alpha_{j}}\leq \sqrt{d}$. Minimizing the function $f(\boldsymbol{\alpha})=\sum_{j}\cos{\alpha_{j}}$, we show that  $f(\boldsymbol{\alpha})$ is always greater than $\sqrt{d}$ in the region (\ref{Region}), hence $T$ cannot be smaller than $T_{low}$. First, we find the critical points of the function $f(\boldsymbol{\alpha})$ inside the region (not on the boundary). Taking the first derivatives of the function in $\alpha_{i}$, we obtain the following  equations:
\begin{equation}
    \sin{\alpha_{i}}=0 \,\,\, \forall i
\end{equation}
This shows that $\alpha_{i}=K_{i}\pi$ and $K_{i}\geq 0$. For these values, $\cos{\alpha_{i}}$ is either $1$ or $-1$, thus the minimum of the function (among these critical points) occurs when we have maximum number of $-1$ which with respect to the constraint (\ref{Region}), $\lfloor \frac{d-1}{4} \rfloor$ number of $\alpha_{i}$ must be equal to $\pi$ and the others be zero. Therefore the minimum is $d-2\lfloor\frac{d-1}{4}\rfloor$ if $\frac{d-1}{4}$ is not an integer. In the case $\frac{d-1}{4}$ is an integer, the point will be on the boundary of the region which we will consider it in the following.

Now, we find the critical points on the boundary of the region (\ref{Region}) where we have $\sum_{j}\alpha_{j}=(d-1)\frac{\pi}{4}$ and $\alpha_{j}\geq 0$. Generally, we assume that we are on the part of the boundary where $x$ number of the $\{\alpha_{i}\}_{i=1}^{d-1}$ are zero. Applying the Lagrange multipliers method, we end up with the equations below:
\begin{equation}
    \sin{\alpha_{i}}=k \,\,\, \forall i,
    \label{Eq-Boundary}
\end{equation}
where $k$ is the Lagrange multiplier. Eqs. (\ref{Eq-Boundary}) show that either $\alpha_{i}=\lambda+2K_{i}\pi$ or $\alpha_{i}=\pi-\lambda +2K_{i}^{'}\pi$ in which $0\leq\lambda\leq \frac{\pi}{2}$ and $K_{i},K_{i}^{'}$ are non-negative integers (because $\alpha_{i}\geq 0$). Being on the part of the boundary with $x$ number of $\alpha_{i }$ to be zero and assuming that $N$ number of them are of the form $\alpha_{i}=\pi-\lambda +2K_{i}^{'}\pi$, we must have (by $\sum_{j}\alpha_{j}=(d-1)\frac{\pi}{4}$):
\begin{equation}
    (d-x-2N)\lambda+(N+\sum_{j} K'_{j}+\sum_{l} K_{l})\pi=(d-1)\frac{\pi}{4}.
\end{equation}
We define $K\equiv \sum_{j} K'_{j}+\sum_{l} K_{l}$. If we write $\lambda$ in terms of $K$ and $N$ we obtain:
\begin{equation}
    \lambda=\frac{(d-1)/2-2(N+K)}{d-x-2N}\frac{\pi}{2}
\end{equation}
and the function takes the form $x+(d-x-2N) \cos{\lambda}$. If we are in the domain $N<\frac{d-x}{2}$ then the function takes its minimum when $\lambda$ is largest and it occurs for $K=0$ (for any $x$ and $N$) . If we are in the domain $N > \frac{d-x}{2}$ then we have:
\begin{equation}
     \lambda=\frac{N-(d-1)/4}{N-(d-x)/2}\frac{\pi}{2}+\frac{K}{N-(d-x)/2}\frac{\pi}{2}.
     \label{Lambda-2}
\end{equation}
Since $x\leq \sqrt{d}$ (otherwise $\sum_{j}\cos{\alpha_{j}}>\sqrt{d}$ and the proof would be done), we can easily show that the first term in Eq. (\ref{Lambda-2}) is greater than $\pi/2$ as the coefficient $\frac{N-(d-1)/4}{N-(d-x)/2}$ is greater than $1$:
\begin{equation}
    N-\frac{d-1}{4}\geq N-\frac{d-\sqrt{d}}{2}\geq N-\frac{d-x}{2}\Longleftrightarrow (\sqrt{d}-1)^2\geq 0.
    \label{Domain-N}
\end{equation}
Also, The second term in \ref{Lambda-2} is positive. Thus, in the domain $N >\frac{d-x}{2}$, $\lambda$ is greater than $\frac{\pi}{2}$ which is a contradiction to the initial assumption $\lambda \leq \frac{\pi}{2}$. Furthermore in the case $N=\frac{d-x}{2}$, from the Eq. \ref{Lambda-2}, we get $d+1+2K=2x$ which is a contradiction as $x$ is a positive integer and $x\leq \sqrt{d}$. Therefore, $N< \frac{d-x}{2}$ and  $\lambda$ takes the following form for the minimum of the function:
\begin{equation}
    \lambda=\frac{(d-1)/4-N}{(d-x)/2-N}\frac{\pi}{2}.
\end{equation}
Moreover, from the Eq. \ref{Domain-N}, we know that $\frac{d-1}{4}\leq \frac{d-x}{2}$ so we must have $0\leq N\leq \frac{d-1}{4}$ because $\lambda\geq 0$. Now, we should see which value of $N$ in the domain minimizes the function. We should obtain the minimum of the function below while $N$ varies:
\begin{equation}
    x+(d-x-2N)\cos{(\frac{(d-1)/4-N}{(d-x)/2-N}\frac{\pi}{2})}.
\end{equation}
By taking the first derivative of this function in $N$ we can easily see that it is monotonically decreasing in the valid domain of $N$, hence the value $N_0=\lfloor(d-1)/4\rfloor$ achieves the minimum of $f(\boldsymbol{\alpha})$ with the value of $ x+(d-x-2\lfloor(d-1)/4\rfloor)\cos{(\frac{(d-1)/4-\lfloor(d-1)/4\rfloor}{(d-x)/2-\lfloor(d-1)/4\rfloor}\frac{\pi}{2})}$ which is always greater than $\sqrt{d}$ for $ d\geq2$:

\begin{multline}
\sqrt{d} \leq x(1-\cos{(\frac{(d-1)/4-N_0}{(d-x)/2-N_0}\frac{\pi}{2})})\\
+ \frac{d+1}{2}\cos{(\frac{1}{\frac{d+1}{2}-\frac{\sqrt{d}}{2}}\frac{\pi}{2})}
\\
\leq 
x+(d-x-2\lfloor(d-1)/4\rfloor)\cos{(\frac{(d-1)/4-\lfloor(d-1)/4\rfloor}{(d-x)/2-\lfloor(d-1)/4\rfloor}\frac{\pi}{2})}
\end{multline}

where for obtaining the second inequality we used the facts that $ 1\leq x\leq \sqrt{d}$ and $\frac{d-1}{4}-1\leq\lfloor \frac{d-1}{4}\rfloor \leq \frac{d-1}{4}$. Thus, the minimum of the function $f(\boldsymbol{\alpha})$ in the region (\ref{Region}) is always greater than $\sqrt{d}$ which is a contradiction to Eq. (\ref{eq:Constraint}), and the proof is complete.

We will now present a lower bound for the speed limit in the Hilbert space of the dimension $d=6$. We will show that the minimal time for transformation of the basis $\{\ket{i}\}_{i=0}^{5}$ to an unbiased basis via a Hamiltonian with fixed mean energy $E$ is bounded below by 
\begin{equation}
    \frac{1}{3E}\arccos{\left(-\frac{\sqrt{6}-4}{2}\right)} \leq T.
\end{equation}

To prove the lower bound, let assume there exist a Hamiltonian for which 
\begin{equation}
T<\frac{1}{3E}\arccos{(-\frac{\sqrt{6}-4}{2})},
\end{equation}
thus we must have $\sum_{i=0}^{5}E_{i}T<2\arccos{(-\frac{\sqrt{6}-4}{2})}$. We define $E_{i}T=\alpha_{i}$ and without loss of generality we consider the minimum eigenenergy of the Hamiltonian $E_{min}=E_{0}=0$. By Eq. (\ref{eq:Constraint}) we must have $-\sqrt{6} \leq \sum_{j}\cos{\alpha_{j}}\leq \sqrt{6}$. We show that the function $f(\boldsymbol{\alpha})=\sum_{j}\cos{\alpha_{j}}$ is always greater than $\sqrt{d}$ in the region
\begin{equation}
    R=\{\sum_{i=0}^{5}\alpha_{i}<2\arccos{(-\frac{\sqrt{6}-4}{2})} \wedge \alpha_{i}>0\textrm{  }\forall i\}.
    \label{Region-d=6}
\end{equation}
Hence, $T$ cannot be smaller than $\frac{1}{3E}\arccos{(-\frac{\sqrt{6}-4}{2})}$.

We minimize  the function $f(\boldsymbol{\alpha})=\sum_{i=0}^{5}\cos{\alpha_{i}}$ in the region closure of $R$. First, we find all the critical points inside the region. By taking the derivatives of the function $f(\boldsymbol{\alpha})$ and equating them to zero, we obtain the critical points as $\alpha_{i}=K_{i}\pi$, $K_{i}\geq 0$ and $K_{i}$ are integers. As $\cos{(K_{i}\pi)}=\pm 1$, the minimum of the function among these critical points occurs when we have the maximum number of $-1$ (with respect to our region $R$, we are allowed to have only one $-1$). Thus the minimum among these critical points is $4$. Now, we find the minimum on the boundaries $\sum_{i=0}^{5}\alpha_{i}=2\arccos{(-\frac{\sqrt{6}-4}{2})}$. Let us assume (without loss of generality) that we are on the part of these boundaries such that $x$ number of $\alpha_{i}$ are zero. Note that $0 \leq x \leq 3$ otherwise the function $f(\boldsymbol{\alpha})$ is greater than $\sqrt{6}$ and we are done with the proof according to Eq.~(\ref{eq:Constraint}). Applying Lagrange multiplier method, we obtain the following set of equations:
\begin{equation}
    \sin{\alpha_{i}}=k, \textrm{  }\forall k
\end{equation}
where $k$ is the multiplier. From these equations we find that $\alpha_{i}$ must be of the following form:
\begin{equation}
    \alpha_{i}=
    \begin{cases}
    \lambda+2K_{i}\pi \,\,\,\,\mathrm{or}\\
    \pi-\lambda + 2K^{\prime}_{i}\pi,
    \end{cases}
    \label{Alpha-Form-d=6}
\end{equation}
in which $0\leq\lambda \leq \frac{\pi}{2}$ and $K_{i}$ and $K^{\prime}_{i}$ are non-negative integers (they must be non-negative as $\alpha_{i}$ are non-negative). We further assume (without loss of generality) that $N$ number of $\alpha_{i}$ are in the second form of Eq. (\ref{Alpha-Form-d=6}). By the constraint on the border of the closure of $R$, we have:
\begin{equation}
    (6-x-N)\lambda + N(\pi-\lambda) +2(\sum_{i}K_{i}+\sum_{j}K^{\prime}_{j})=2\arccos{-\frac{\sqrt{6}-4}{2}}.
\end{equation}
Solving this equation for $\lambda$ we obtain:
\begin{equation}
    \lambda=\frac{2\arccos{(-\frac{\sqrt{6}-4}{2})}-(2K+N)\pi}{6-x-2N}.
    \label{Angle-Lagr-Multiplier}
\end{equation}
where $K=\sum_{i}K_{i}+\sum_{j}K^{\prime}_{j}$. Eq. (\ref{Angle-Lagr-Multiplier}) implies that $N<\frac{6-x}{2}$ otherwise $\lambda>\pi/2$ which is a contradiction (to the initial assumption that $0\leq\lambda \leq \frac{\pi}{2}$). The function $f(\boldsymbol{\alpha})$ for the critical points on the boundary becomes $(6-x-2N)\cos{(\frac{2\arccos{(-\frac{\sqrt{6}-4}{2})}-(2K+N)\pi}{6-x-2N})}$. Considering that $N<\frac{6-x}{2}$ and $0\leq \lambda \leq \frac{\pi}{2}$, it takes its minimum for any $x$ and $N$ when $K=0$. Thus the minimum of the function on the boundary must be of the form $(6-x-2N)\cos{(\frac{2\arccos{(-\frac{\sqrt{6}-4}{2})}-N\pi}{6-x-2N})}$ which is greater than or equal $\sqrt{6}$ for any $1\leq x \leq 3$ and $N<\frac{6-x}{2}$. Therefore, the minimum of the function $f(\boldsymbol{\alpha})$ over the region $R$ is greater than $\sqrt{6}$ which is a contradiction to Eq.~(\ref{eq:Constraint}) and the proof is complete.

\end{document}